\documentclass[11pt]{article}

\usepackage[utf8]{inputenc}
\usepackage[T1]{fontenc}
\usepackage[margin=1in]{geometry}
\usepackage{amsmath,amssymb,amsthm}
\usepackage{booktabs}
\usepackage{graphicx} 
\usepackage{xcolor}
\usepackage{eso-pic}
\usepackage{tikz}
\usetikzlibrary{arrows.meta, positioning}
\usepackage[hidelinks]{hyperref}

\theoremstyle{plain}
\newtheorem{proposition}{Proposition}
\newtheorem{corollary}{Corollary}
\theoremstyle{definition}
\newtheorem{assumption}{Assumption}

\newcommand{\PhaseZeroSpend}{5.13}
\newcommand{\PZaScored}{89}
\newcommand{\PZaAdr}{0\%}
\newcommand{\PZaOracleAdr}{44\%}
\newcommand{\PZaElasticity}{0.000}
\newcommand{\PZaRegret}{2520}
\newcommand{\PZbScored}{5}
\newcommand{\PZbElasticity}{0.976}
\newcommand{\PZcScored}{100}
\newcommand{\PZcEpisodes}{100}
\newcommand{\PZcAdr}{42\%}
\newcommand{\PZcOracleAdr}{43\%}
\newcommand{\PZcElasticity}{0.992}
\newcommand{\PZcMarkerAuc}{0.505}
\newcommand{\PZcJudgeAuc}{0.500}
\newcommand{\PZcFlags}{0\%}

\newcommand{\ResHoneFull}{SUPPORTED}
\newcommand{\HoneFullMetaAdr}{60\%}
\newcommand{\ResHoneLadder}{SUPPORTED}

\newcommand{\HoneLadderFlag}{0\%}
\newcommand{\HoneLadderAuc}{0.50}
\newcommand{\HoneLadderAdrHaiku}{60\%}
\newcommand{\HoneLadderAdrFable}{62\%}
\newcommand{\ResHtwo}{REFRAMED}
\newcommand{\HtwoCriticMeta}{25\%}
\newcommand{\HtwoGateMeta}{50\%}
\newcommand{\HtwoDetCrossC}{100\%}
\newcommand{\HtwoDetCrossG}{100\%}
\newcommand{\HtwoDetDupC}{0\%}
\newcommand{\HtwoDetDupG}{0\%}
\newcommand{\HtwoDetMissingC}{72\%}
\newcommand{\HtwoDetMissingG}{100\%}
\newcommand{\HtwoDetOutlierC}{0\%}
\newcommand{\HtwoDetOutlierG}{0\%}
\newcommand{\HtwoDetSchemaC}{0\%}
\newcommand{\HtwoDetSchemaG}{100\%}
\newcommand{\HtwoDetUnitC}{0\%}
\newcommand{\HtwoDetUnitG}{0\%}
\newcommand{\HtwoDetStaleC}{0\%}
\newcommand{\HtwoDetStaleG}{100\%}
\newcommand{\HtwoDetSupersededC}{100\%}
\newcommand{\HtwoDetSupersededG}{100\%}
\newcommand{\ResHthree}{NOT SUPPORTED}
\newcommand{\HthreeGateDet}{62\%}
\newcommand{\HthreeCriticDet}{31\%}
\newcommand{\HthreeGateResid}{294}
\newcommand{\HthreeCriticResid}{311}
\newcommand{\HthreeFalseBlock}{0\%}
\newcommand{\ResHfour}{NOT SUPPORTED}
\newcommand{\HfourRecovery}{-0.04}
\newcommand{\HfourStaleLossA}{134}
\newcommand{\HfourStaleLossD}{-0.5}

\newcommand{\LadderEndpointDiff}{0.0125}
\newcommand{\LadderEndpointCILo}{-0.0551}
\newcommand{\LadderEndpointCIHi}{0.0801}
\newcommand{\LadderTrendPerTier}{0.0040}
\newcommand{\HtwoDetN}{100}
\newcommand{\HtwoCICrossCLo}{96}
\newcommand{\HtwoCICrossCHi}{100}
\newcommand{\HtwoCICrossGLo}{96}
\newcommand{\HtwoCICrossGHi}{100}
\newcommand{\HtwoCIDupCLo}{0}
\newcommand{\HtwoCIDupCHi}{4}
\newcommand{\HtwoCIDupGLo}{0}
\newcommand{\HtwoCIDupGHi}{4}
\newcommand{\HtwoCIMissingCLo}{63}
\newcommand{\HtwoCIMissingCHi}{80}
\newcommand{\HtwoCIMissingGLo}{96}
\newcommand{\HtwoCIMissingGHi}{100}
\newcommand{\HtwoCIOutlierCLo}{0}
\newcommand{\HtwoCIOutlierCHi}{4}
\newcommand{\HtwoCIOutlierGLo}{0}
\newcommand{\HtwoCIOutlierGHi}{4}
\newcommand{\HtwoCISchemaCLo}{0}
\newcommand{\HtwoCISchemaCHi}{4}
\newcommand{\HtwoCISchemaGLo}{96}
\newcommand{\HtwoCISchemaGHi}{100}
\newcommand{\HtwoCIUnitCLo}{0}
\newcommand{\HtwoCIUnitCHi}{4}
\newcommand{\HtwoCIUnitGLo}{0}
\newcommand{\HtwoCIUnitGHi}{4}
\newcommand{\HtwoCIStaleCLo}{0}
\newcommand{\HtwoCIStaleCHi}{4}
\newcommand{\HtwoCIStaleGLo}{96}
\newcommand{\HtwoCIStaleGHi}{100}
\newcommand{\HtwoCISupersededCLo}{96}
\newcommand{\HtwoCISupersededCHi}{100}
\newcommand{\HtwoCISupersededGLo}{96}
\newcommand{\HtwoCISupersededGHi}{100}
\newcommand{\BootReplicates}{10000}
\newcommand{\BootSeed}{20260728}
\newcommand{\HthreeGateResidLo}{104.1}
\newcommand{\HthreeGateResidHi}{498.0}
\newcommand{\HthreeCriticResidLo}{120.3}
\newcommand{\HthreeCriticResidHi}{516.3}
\newcommand{\HfourStaleLossALo}{-31.9}
\newcommand{\HfourStaleLossAHi}{305.2}
\newcommand{\HfourStaleLossDLo}{-3.7}
\newcommand{\HfourStaleLossDHi}{1.9}

\newcommand{\ConvMAE}{0.015}
\newcommand{\ConvR}{0.876}
\newcommand{\ConvCICoverage}{15/16}
\newcommand{\CovRecovery}{-0.0411}
\newcommand{\CovResidual}{0.000000}
\newcommand{\CovDominant}{\texttt{silent\_unit\_change}}
\newcommand{\CovDominantSharePct}{83}
\newcommand{\LooFull}{-0.0411}
\newcommand{\LooTopClass}{\texttt{silent\_unit\_change}}
\newcommand{\LooTopRecovery}{0.6219}
\newcommand{\ThreshReg}{0.005}
\newcommand{\ThreshRegAdr}{0.6175}
\newcommand{\ThreshLoAdr}{0.6225}
\newcommand{\ThreshHiAdr}{0.3325}
\newcommand{\StressIntact}{0.8333}
\newcommand{\StressAllDegraded}{0.6667}
\newcommand{\StressFragile}{\texttt{stale\_master\_data}}
\newcommand{\FalsGate}{0.2715}
\newcommand{\FalsRandom}{0.1636}
\newcommand{\FalsShuffled}{0.1336}
\newcommand{\FalsAnti}{0.1173}
\newcommand{\FalsGap}{0.1079}
\newcommand{\SecondADR}{0.3146}
\newcommand{\SecondRecovery}{1.0000}
\newcommand{\SecondExpiredADR}{0.9313}
\newcommand{\SecondConsentADR}{0.9313}

\newcommand{\ci}[2]{~{\scriptsize[#1,\,#2]}}

\title{\textbf{SARC-DQ: Runtime Data-Quality Gating for Agentic AI}\\
  \large Silent evidence defects, the incompetence shield, and downstream-only remediation}
\author{Gaston Besanson\\ Universidad Torcuato Di Tella}
\date{Working paper --- \today}

\begin{document}
\maketitle

\begin{abstract}
\textbf{Problem.} Agentic systems act, so a defect in the evidence they retrieve becomes
a wrong action with a currency cost. \textbf{Gap.} The most dangerous enterprise defects
are \emph{metadata-borne}---a stale price, a superseded record---perfectly well-formed in
the payload and betrayed only by freshness, lineage, or provenance. Such a defect never
enters the agent's context, so a payload-only reader cannot detect it regardless of model
scale. \textbf{Idea.} We argue that the fix is architectural: a cheap metadata-aware gate
placed \emph{at the point of action} enforces evidence quality where a stronger model
cannot, remediating \emph{downstream-only} (quarantine-and-substitute in a governed buffer;
zero writes to source). \textbf{Evidence.} On a priced replenishment task we measure the
\emph{silent conversion} of a metadata-borne defect into a loss against a same-seed clean
counterfactual, across a four-model tier ladder and eight defect classes.
\textbf{Result.} Across four model tiers spanning approximately $15\times$ in inference
price (haiku$\to$sonnet$\to$opus$\to$fable), metadata-borne action-defect rates remain flat
to slightly rising, while behavioral doubt markers stay at chance (AUC
$\le\HoneLadderAuc$) and explicit data-quality flags remain absent ($\HoneLadderFlag$). A
metadata-aware Pre-Action Gate outperforms a realistic payload-only critic on the freshness
and schema signals it covers; uncovered unit-consistency and plausible-outlier defects
remain unresolved, and portfolio-level recovery is not achieved. A model-free analytical
oracle derived from the task's decision geometry tracks the measured agent ADR with MAE
\ConvMAE{} and Pearson $r=\ConvR$, the measured interval covering the prediction in
\ConvCICoverage{} cells, giving the flat ladder an analytical form. \textbf{Impact.} These
results identify evidence integrity as a distinct systems axis from model capability, and
show that mitigation effectiveness depends on enforcement placement and predicate coverage.
\end{abstract}

\section{Introduction}
An agent differs from a chatbot in one decisive way: it \emph{acts}. A chatbot that
reasons over a stale price emits a wrong sentence; an agent that reasons over a
stale price places a wrong order, and the loss is realised in currency. The
literature on agent robustness has studied degradation under generic noise, but the
defects that matter most in enterprise data are not noisy---they are
\emph{plausible}. A price that is 120 days stale is a perfectly well-formed number;
nothing in its payload betrays it. The tell is in the \emph{metadata}: freshness,
lineage, provenance, version. We call such defects \emph{metadata-borne}, and
contrast them with \emph{payload-visible} defects (duplicates, contradictions,
schema drift) detectable from content alone.

\paragraph{The silent conversion chain.} A metadata-borne defect propagates:
\[
\text{data defect} \;\rightarrow\; \text{(invisible in payload)} \;\rightarrow\;
\text{agent acts} \;\rightarrow\; \text{action defect} \;\rightarrow\; \text{loss}.
\]
Each arrow is measurable. We price the last one against a same-seed clean
counterfactual, so ``loss'' is a currency delta, not a task-success proxy
(Fig.~\ref{fig:chain}).

\begin{figure}[t]
\centering
\begin{tikzpicture}[
  >={Latex[length=2mm]},
  box/.style={draw, rounded corners, align=center, minimum height=7mm,
              text width=17mm, inner sep=2pt, font=\scriptsize},
  bad/.style={box, fill=red!7},
  good/.style={box, fill=green!8},
  gate/.style={box, fill=blue!8, very thick, text width=20mm},
  lbl/.style={font=\scriptsize\itshape, align=center},
]
\node[bad] (rec)   {stale price\\ \tiny(payload valid)};
\node[bad, right=9mm of rec]   (agent) {agent acts\\confidently};
\node[bad, right=9mm of agent] (order) {wrong order};
\node[bad, right=9mm of order] (loss)  {currency loss};
\draw[->] (rec)   -- (agent);
\draw[->] (agent) -- (order);
\draw[->] (order) -- (loss);
\node[lbl, above=1mm of agent] {invisible in payload};
\node[lbl, above=1mm of order] {no expressed doubt};
\node[gate, below=11mm of agent] (gate) {Pre-Action Gate\\ \tiny(metadata-aware)};
\node[good, right=13mm of gate] (corr) {corrected order};
\draw[->] (rec) |- (gate);
\draw[->] (gate) -- node[lbl, above]{detect + substitute} (corr);
\end{tikzpicture}
\caption{The \emph{silent conversion} of a metadata-borne defect (top): the payload is
valid, so a payload-only agent acts confidently and converts the defect into a currency
loss. A metadata-aware Pre-Action Gate (bottom) intercepts the same record where the
signal actually lives---the metadata---and substitutes a governed value, yielding a
corrected action. Model scale moves along the top row; enforcement placement adds the
bottom one.}
\label{fig:chain}
\end{figure}
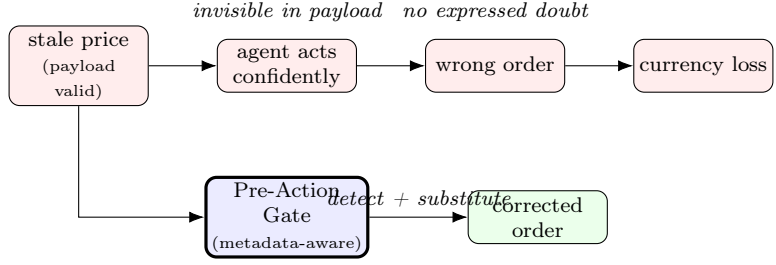

\paragraph{Competence buys conversion, not detection: the incompetence shield.}
One might hope a smarter agent would notice the bad price. We observe that it does
not---the discriminating signal is absent from the payload context at every tested model
tier. Worse,
making the agent \emph{better at its task} makes it \emph{more} vulnerable: an agent
that faithfully tracks the price it is given (high decision elasticity) inherits the
corruption, while an agent that ignores the price (low elasticity) is accidentally
shielded by its own incompetence---the \emph{incompetence shield}. We therefore argue
that detection must be placed where the signal lives---on the metadata, at the point of
action---not bought with model scale.

\paragraph{Three claims.} Everything below serves three messages, in order of
importance. \textbf{(M1)}~Metadata-borne defects are \emph{structurally} invisible to a
payload-only agent (we argue this from the architecture and show it in the channel
measurements). \textbf{(M2)}~Moving to higher tested model tiers does not fix M1---silence
and loss-conversion are flat across the four-model tier ladder (we observe this,
\S\ref{sec:results}; the paper's headline). \textbf{(M3)}~The correct remedy is architectural: runtime
metadata enforcement at the point of action (we argue this and give a gate that beats a
realistic payload-only critic on its covered channel). The benchmark, formal invariants,
economics, and governance mapping are \emph{support} for these three, not competing
claims.

\paragraph{Two kinds of contribution.} We separate what is scientific from what is a
research artifact, so each can be judged on its own terms.
\begin{table}[t]\centering\small
\setlength{\tabcolsep}{5pt}
\begin{tabular}{@{}p{0.30\linewidth}p{0.34\linewidth}p{0.28\linewidth}@{}}
\toprule
\textbf{Contribution} & \textbf{Evidence / where} & \textbf{Validation} \\
\midrule
\multicolumn{3}{@{}l}{\emph{Scientific}}\\
Silent conversion of metadata-borne defects (M1) & Channel model + H2 measurement (\S\ref{sec:arch},\,\ref{sec:results}) & Pre-registered; measured \\
Capability does not buy skepticism (M2) & Four-model ladder, H1 (\S\ref{sec:results}) & Pre-registered; measured \\
Enforcement-placement thesis (M3) & Gate vs.\ realistic critic, H3/H4 (\S\ref{sec:results}) & Argued; partially measured \\
Incompetence shield & Phase 0 elasticity contrast (\S\ref{sec:pilot}) & Measured (pilot) \\
\midrule
\multicolumn{3}{@{}l}{\emph{Research artifact}}\\
SARC-DQ predicate family + Pre-Action Gate & \S\ref{sec:arch}; \texttt{src/sarc\_dq/} & Typed, tested, mock+live \\
Downstream-only remediation with lineage & Architectural invariants (\S\ref{sec:formal}) & Proved (Prop.~\ref{prop:lineage}) \\
GIGO-Bench (frozen, priced benchmark) & \S\ref{sec:gigo} & \texttt{make gigo-verify} \\
Reproducible, provenance-linked pipeline & \S\ref{sec:repro}; \texttt{EXPERIMENT\_STATUS.md} & Generated, deterministic \\
\bottomrule
\end{tabular}
\caption{Scientific claims vs.\ the research artifact. Every paper value is generated
from a committed \texttt{results/<exp>-live} summary; nothing is hand-entered.}
\label{tab:contrib}
\end{table}

\section{Related work}
\paragraph{Agent robustness under noise.} AgentNoiseBench~\cite{agentnoisebench}
finds tool-side noise degrades agents more than user-side noise, and that strong
reasoning models make \emph{confident} errors under corrupted tool feedback;
NoisyToolBench~\cite{noisytoolbench}, ToolEmu~\cite{toolemu}, and
``Tools Fail''~\cite{toolsfail} study silent tool errors;
$\tau$-bench / $\tau^2$-bench~\cite{taubench} provide agentic task
environments. These works inject \emph{generic} noise and measure \emph{task-success}
deltas. Our delta is precise: we isolate the \emph{metadata-borne} channel
(invisible-in-payload defects), price failures in \emph{currency} rather than
task-success, and supply and test an \emph{architectural remedy} with downstream-only
remediation.
\paragraph{Action-side runtime enforcement.} AgentSpec- and
MI9-class systems~\cite{agentspec,mi9} and LLM guardrails enforce constraints on what the agent
is about to \emph{do}. SARC-DQ addresses the complementary \emph{input-side} gap:
whether the evidence the action rests on is trustworthy.
\paragraph{Data validation and DQ dimensions.} ISO/IEC 5259~\cite{iso5259}
frames data-quality dimensions; the data-cleaning benchmarks we use for Tier-2
validation---the Raha/Baran suite~\cite{raha}, multi-source-conflict
datasets, Magellan/DeepMatcher~\cite{magellan} entity matching, and ALFRED
archival vintages~\cite{alfred}---supply labeled real errors.
\paragraph{Taxonomy grounding.} The eight corruption classes are not invented: they
map onto the canonical data-quality taxonomies. Rahm and Do~\cite{rahmdo2000} split
errors along single/multi-source $\times$ schema/instance --- our metadata-borne
classes are the single-source instance-level cell; cross-source contradiction and
duplicate-conflicting-terms are the multi-source cell. Kim et al.~\cite{kim2003}
enumerate 33 dirty-data types; each of our classes maps to a Kim type. The predicate
dimensions (freshness, completeness, consistency, lineage) follow Wang and
Strong~\cite{wangstrong1996} and ISO~8000~\cite{iso8000}; Sambasivan et
al.~\cite{sambasivan2021} motivate why upstream data defects cascade into
high-stakes AI failure. \textbf{Prevalence.} Aggregate real-world corruption is
high: 47\% of newly created records carry at least one critical, work-impacting
error and only 3\% of data-quality scores are acceptable~\cite{nagle2017};
organizations self-estimate 17--32\% of data inaccurate~\cite{experian2017}; and
sources conflict on $\sim$70\% of items on the deep web~\cite{li2013}. GIGO's
sweep rates \{2, 5, 10, 20\%\} therefore sit at or below this band --- 20\% is its
low end, 2--10\% model well-governed single sources. \textbf{No injector parameter
is arbitrary}: each is computed from public labeled data (Raha/Baran~\cite{mahdavi2019},
Magellan/DeepMatcher~\cite{konda2016}, Dong Stock/Flight~\cite{li2013},
ALFRED~\cite{alfred} /
schema-evolution vintages~\cite{curino2013}), taken from cited literature, or declared
as a flagged default; the provenance of every value is in
\texttt{src/sarc\_dq/specs/taxonomy\_v1\_calibrated.yaml} and
\texttt{benchmarks/gigo/CALIBRATION.md}.

\section{Architecture}\label{sec:arch}
SARC-DQ places a data-quality predicate family on the four SARC enforcement sites
(Pre-Action Gate, Action-Time Monitor, Post-Action Auditor, Escalation Router). An
evidence record has two channels: a \emph{payload} (the field values a naive
consumer sees) and \emph{metadata} (freshness, lineage, provenance, version). A
corruption class declares a \emph{channel}: payload-visible or metadata-borne. This
tag is load-bearing---a payload-only critic structurally cannot see a metadata-borne
defect.

\paragraph{Predicates and responses.} Six parameterized predicates---%
\texttt{freshness(max\_age)}, \texttt{lineage\_present},
\texttt{golden\_record\_unique}, \texttt{cross\_source\_consistent(tol)},
\texttt{schema\_conformant}, \texttt{complete(required)}---are authored as a YAML
constraint spec (class $\in$ hard/soft/escalation; verification point; response
protocol; operating point). Responses are \emph{block}, \emph{degrade autonomy},
\emph{escalate}, or \emph{quarantine-and-substitute}.

\paragraph{Downstream-only remediation.} On \emph{quarantine-and-substitute}, the
gate replaces the offending value with a known-good value drawn \emph{only} from a
governed buffer---a downstream, versioned store keyed by content-addressed evidence
IDs. The gate is read-only over the evidence set and \textbf{never writes to any
source store}. Every admitted action logs the versioned evidence set it relied on,
giving full lineage from action back to the exact records.

\section{Architectural invariants}\label{sec:formal}
We state two properties of the architecture---not theorems we advertise as a
contribution, but the guarantees a reviewer should be able to check. The first is a
\emph{correctness guarantee} (remediation cannot corrupt lineage); the second is an
\emph{analytical decomposition} showing precisely which term runtime gating multiplies,
and why model capability does not appear in it.

\begin{assumption}\label{as:addr}
Each evidence record has a content-addressed identifier $\mathrm{eid}(r)=H(\text{payload}(r),\text{metadata}(r))$
for a collision-resistant $H$. Remediation creates new buffer records and never
mutates a source record.
\end{assumption}

\begin{proposition}[Lineage preservation under downstream-only remediation]\label{prop:lineage}
Under Assumption~\ref{as:addr}, for every admitted action $a$ there is a recorded
evidence set $E(a)$ whose identifiers resolve to the exact records $a$ relied upon,
and every source record referenced by $E(a)$ is bit-identical to its pre-action
state.
\end{proposition}
\begin{proof}
The gate logs $E(a)=\{\mathrm{eid}(r):r\in \text{evidence}(a)\}$ at admission. A
substitution introduces a new buffer record $r'$ with a fresh identifier
$\mathrm{eid}(r')$ and appends it to $E(a)$; it performs no write to any $r$ in a
source store (Assumption~\ref{as:addr}). Hence every $\mathrm{eid}\in E(a)$
resolves, by collision-resistance of $H$, to the unique record whose payload and
metadata hash to it; and since no source record was written, each such source
record equals its pre-action state. Therefore the map $a\mapsto E(a)\mapsto$
records is total and state-preserving. \qed
\end{proof}

\begin{proposition}[Evidence-error-to-ADR bound]\label{prop:bound}
Let $\varepsilon$ be the probability that an action's evidence carries a material
defect, and let $\rho\in[0,1]$ \emph{upper-bound} the conditional probability that such a
defect induces a material action deviation. Let $d$ be the gate's detection probability on
the defect's channel and $s$ the probability that a detected defect is successfully
substituted. Under the stated independence (defect independent of the demand realisation)
and non-interference (detection/recovery act only on detected defects) assumptions,
\[
\mathrm{ADR}_{\text{no gate}} \le \varepsilon\rho, \qquad
\mathrm{ADR}_{\text{gate}} \le \varepsilon\rho\,(1-ds).
\]
\end{proposition}
\begin{proof}[Proof sketch]
A material action deviation requires both a material defect (prob.\ $\varepsilon$) and
conversion (prob.\ at most $\rho$, since $\rho$ upper-bounds the conditional); independence
gives the product $\varepsilon\rho$ as an upper bound. A gate removes a defect from the
loss path when it both detects ($d$) and successfully substitutes ($s$); the surviving
fraction is $1-ds$, scaling the bound. Both statements are upper bounds by construction of
$\rho$; we do not assert equality.
\end{proof}
\noindent\emph{Reading of the decomposition.} Model capability does not appear in
either bound. In the observed task, greater decision elasticity can \emph{raise} $\rho$
because the agent tracks the corrupted evidence more faithfully (the incompetence shield in
reverse); model capability does not, however, create a nonzero metadata-channel detection
term $d$ when that metadata is absent from context. Only \emph{placement}---the gate's $d$
and $s$---multiplies the ceiling by $1-ds$. This is the analytical form of M2 and M3:
in this setting scale moves $\rho$ the wrong way, while enforcement moves the factor that
matters.

\noindent\emph{Honesty note.} These are \emph{upper bounds}, not equalities: $\rho$
is defined as an upper bound and is not constant across episodes---the effective
conversion is elasticity-weighted, which
the pilot estimates empirically (elasticity $\PZcElasticity$ for the competent
agent). We do not claim a tight closed form; the bound's value is qualitative.

\begin{corollary}[Closed-form conversion is model-independent]\label{cor:conv}
On the newsvendor substrate the conversion factor $\rho$ admits an explicit oracle form: for
a corrupted episode the agent that acts on the newsvendor optimum for the \emph{shown} price
$\hat c$ incurs a paired loss $\ell = C(q^\star(\hat c)) - C(q^\star(c))$ against the same
realised demand, and the episode is material iff $\ell \ge \tau_m\,C(q^\star(c))$. The
per-cell mean of this indicator is a closed-form upper estimate of $\rho$ that \emph{contains
no model term}: the shown price $\hat c$ is fixed by the injected corruption, not by the
agent, so model tier cannot enter---metadata availability, not capability, is what is held
fixed. Appendix~\ref{app:analysis} shows this oracle conversion tracks the measured agent ADR
to a mean absolute error of $\ConvMAE$ (Pearson $r=\ConvR$; measured CI covers the prediction
in $\ConvCICoverage$ cells), which is the analytical form of M2.
\end{corollary}
\noindent The corollary is stated as an estimate of the same upper bound, not a tighter
equality; the residual cell is reported honestly in the appendix rather than fitted away.

\section{The Phase 0 pilot}\label{sec:pilot}
We ran a three-run live pilot on a priced single-period (newsvendor) replenishment
task, where the cost-minimising order depends on the unit price; a stale price
yields a wrong order and a currency loss versus the same-seed clean counterfactual.
Agent \texttt{claude-sonnet-5}, judge \texttt{claude-haiku-4-5}; total spend
\$\PhaseZeroSpend.

\begin{table}[t]\centering
{\small\setlength{\tabcolsep}{4.5pt}
\begin{tabular}{lcccccc}
\toprule
run & prompt & scored & verdict & agent-ADR & oracle-ADR & elasticity \\
\midrule
0a & naive & \PZaScored/100 & AMBIGUOUS & \PZaAdr & \PZaOracleAdr & \PZaElasticity \\
0b & policy & \PZbScored/100 & \emph{rejected} & --- & --- & \PZbElasticity \\
0c & policy & \PZcScored/\PZcEpisodes & \textbf{SUPPORTED} & \PZcAdr & \PZcOracleAdr & \PZcElasticity \\
\bottomrule
\end{tabular}}
\caption{Phase 0 pilot (real live numbers, generated from the results branches).
Behavioral marker/judge AUC on 0c: \PZcMarkerAuc/\PZcJudgeAuc; explicit data-flags
\PZcFlags. \emph{0b is diagnostic only}: its elasticity (\PZbElasticity,
$n{=}\PZbScored$ scored pairs) comes from an \textbf{invalid} run --- the
instrumentation failure the Phase 0c validity gate was built to catch --- and is
not evidence; it is reported for methods integrity, never interpreted.}
\end{table}

\paragraph{Silence, and the incompetence shield.} On the valid run (0c) the competent
agent acts on the stale price \PZcAdr{} of the time---matching the metadata-blind
oracle (\PZcOracleAdr)---while its transcripts carry no discriminating signal
(marker AUC \PZcMarkerAuc, judge AUC \PZcJudgeAuc, both at chance; flags \PZcFlags).
The naive agent (0a) converts \PZaAdr: it is \emph{inelastic} (decision elasticity
\PZaElasticity, derived from the 0a log; clean-arm regret \$\PZaRegret{} versus the
oracle), barely tracking the price, and is thus shielded from the defect only by
being bad at the task. Making the agent competent (elasticity \PZcElasticity) made
it vulnerable.

\paragraph{Methods integrity: GIGO in our own pipeline.} Run 0b returned
\texttt{SUPPORTED} on only \PZbScored/100 scored pairs---the agent's replies were
truncated before the decision line, and a brittle parser dropped the rest. A verdict
read from five pairs is meaningless. We report 0b as-is; it motivated a
\emph{validity precondition} (any run scoring $<80/100$ is \texttt{INVALID}
regardless of metrics). We caught garbage-in-garbage-out in our own measurement
instrument---a fitting first result for a data-quality project, reported rather than
hidden.

\section{GIGO-Bench}\label{sec:gigo}
GIGO-Bench freezes a conditions matrix: eight taxonomy-v0 corruption classes
$\times$ four rates $\{2,5,10,20\%\}$ $\times$ six mitigation arms (A no-gate, B
advisory, C payload-only critic, D DQ gate, E oracle, F$(v)$ upstream cleaning),
$=192$ cells. Metrics: ADR; loss (median/P90/P99/mean); detection rate by
class$\times$arm; false-block rate; completion rate; recovery ratio; gate overhead
latency; staleness-bound coverage. Seeds and workload-level splits are frozen;
\texttt{make gigo-verify} checks a re-run against the committed reference within
per-cell tolerances. The checked-in reference is the deterministic \emph{mock}
matrix (a pipeline reference); live numbers land via the Part-4 experiment kits.

\paragraph{Arm information access (defined before results).} What each arm may \emph{see}
determines what it can detect, so we fix it here. \textbf{C} (payload-only critic) sees
record payloads with no freshness/lineage/provenance. \textbf{D} (Pre-Action Gate) evaluates
metadata predicates and, on \texttt{quarantine-and-substitute}, draws a governed-buffer
value. \textbf{E} (oracle) acts on the true value---an upper bound, not deployable.
\textbf{F$(v)$} is an \emph{idealized upstream-cleaning comparator}: for a fraction $v$ of
corrupted episodes it restores the known clean benchmark value, so it has privileged
ground-truth access in the synthetic benchmark and represents an optimistic,
upper-bound-style cleaning mechanism rather than a deployable detector. The defensible
operational comparison for the gate is therefore against the realistic payload-only critic
C, not against E or F$(v)$.

\paragraph{Signed paired loss (defined before results).} Loss is a same-seed
\emph{difference}: the realised cost of the agent's order on the corrupted view minus its
order on the true-price view, sharing one demand draw. Because the shared demand realisation
still enters both realised costs, a corrupted or substituted path can occasionally yield a
marginally \emph{lower} realised cost than the clean-view path, so individual paired losses
can be slightly negative. We retain signed values (no clipping; clipping was not
preregistered) and interpret values near zero by their bootstrap interval: a small negative
mean, or any interval crossing zero, is \emph{no established effect}, never evidence of
systematic improvement or of recovery beyond 100\%.

\section{Results (H1--H4)}\label{sec:results}
All predictions were pre-registered (\texttt{reports/prereg/}) and frozen before any
run; every number below is a macro generated from a committed \texttt{results/<exp>-live}
summary (\texttt{paper/scripts/ingest\_results.py}), never hand-entered. All four
experiments were re-run on the corrected, cap-hardened harness ($\texttt{policy\_instructed}$
decider, paired counterfactual loss, rate-dependent sampler); the invalid first-wave and
one cap-truncated run are retained on their branches for audit and are \emph{not} ingested
(see the deviations note and \texttt{reports/FINDINGS.md}).

\paragraph{H1 (Silence) --- \ResHoneFull.}
Under the competent, metadata-blind decider, a corrupted price is silently converted into a
material order defect on \HoneFullMetaAdr{} of metadata-borne episodes, while the agent's
transcripts carry no discriminating doubt signal. The loss-conversion channel is real and
the decider is silent about it, as predicted (agent-ADR $\ge 20\%$; behavioral AUC $\le 0.60$).

\paragraph{H1 model-tier ladder --- \ResHoneLadder{} (headline).}
Across four product tiers---haiku-4-5\,$\to$\,sonnet-5\,$\to$\,opus-4-8\,$\to$\,fable-5,
increasingly expensive and nominally more capable, spanning $\approx 15\times$ in inference
price---the metadata-borne ADR is flat-to-rising (\HoneLadderAdrHaiku{} at haiku,
\HoneLadderAdrFable{} at fable) while \emph{both} silence signals stay pinned at the null:
behavioral marker AUC $\le \HoneLadderAuc{}$ and explicit-flag fraction \HoneLadderFlag{} at
every rung. \textbf{Capability does not buy skepticism}: a frontier decider seeing only the
payload is exactly as silent, and exactly as vulnerable, as a small one --- the empirical
core of the placement thesis. (The LLM-judge AUC, the second behavioral quantity, is
deferred; the claim rests on the marker AUC and flag fraction.)

\paragraph{H2 (Detection asymmetry) --- \ResHtwo{}.}
The registered clean channel split (C\,$\approx$\,D on payload-visible \emph{and}
C\,$\ll$\,D on metadata-borne) was \emph{not} supported. The honest evidence is per class
(Table~\ref{tab:h2}): detection follows the signal actually exposed to each mechanism and
the predicates actually implemented, not the assigned taxonomy label. The gate strictly
dominates the payload-only critic on freshness (\texttt{stale\_master\_data}: critic
\HtwoDetStaleC{}, gate \HtwoDetStaleG{}) and schema-conformance (\texttt{schema\_drift}:
\HtwoDetSchemaC{}\,$\to$\,\HtwoDetSchemaG{}; \texttt{missing\_mandatory\_field}:
\HtwoDetMissingC{}\,$\to$\,\HtwoDetMissingG{}) --- signals the critic does not observe. Two
labels cross the assigned boundary: \texttt{schema\_drift}, registered as payload-visible,
is missed by the critic and caught by the deterministic schema predicate; and
\texttt{superseded\_golden\_record}, registered as metadata-borne, is \emph{visible} to the
critic (\HtwoDetSupersededC{}) because the companion record creates a payload-level
contradiction. Both mechanisms fail on the two uncovered defects
(\texttt{silent\_unit\_change}, \texttt{plausible\_outlier}: 0\%/0\%), reported as coverage
gaps, not patched. The experiment therefore converts the payload/metadata distinction from
an assigned taxonomy into an \emph{empirically observed detection boundary}: detectability
depends on which representation and companion evidence a mechanism receives, and on the
predicates it implements. The pooled metadata figures (critic \HtwoCriticMeta{}, gate
\HtwoGateMeta{}) are a secondary summary \emph{under the preregistered labels}, and inherit
those labels' incompleteness.

\begin{table}[t]\centering\small
\setlength{\tabcolsep}{5pt}
\begin{tabular}{@{}l l l l@{}}
\toprule
corruption class & critic C & gate D & registered channel \\
\midrule
\texttt{stale\_master\_data}        & \HtwoDetStaleC\ci{\HtwoCIStaleCLo}{\HtwoCIStaleCHi}           & \HtwoDetStaleG\ci{\HtwoCIStaleGLo}{\HtwoCIStaleGHi}           & metadata-borne \\
\texttt{schema\_drift}              & \HtwoDetSchemaC\ci{\HtwoCISchemaCLo}{\HtwoCISchemaCHi}        & \HtwoDetSchemaG\ci{\HtwoCISchemaGLo}{\HtwoCISchemaGHi}        & payload-visible$^\dagger$ \\
\texttt{missing\_mandatory\_field}  & \HtwoDetMissingC\ci{\HtwoCIMissingCLo}{\HtwoCIMissingCHi}     & \HtwoDetMissingG\ci{\HtwoCIMissingGLo}{\HtwoCIMissingGHi}     & payload-visible \\
\texttt{superseded\_golden\_record} & \HtwoDetSupersededC\ci{\HtwoCISupersededCLo}{\HtwoCISupersededCHi} & \HtwoDetSupersededG\ci{\HtwoCISupersededGLo}{\HtwoCISupersededGHi} & metadata-borne$^\dagger$ \\
\texttt{cross\_source\_contradiction} & \HtwoDetCrossC\ci{\HtwoCICrossCLo}{\HtwoCICrossCHi}         & \HtwoDetCrossG\ci{\HtwoCICrossGLo}{\HtwoCICrossGHi}           & payload-visible \\
\texttt{duplicate\_vendor\_conflicting\_terms} & \HtwoDetDupC\ci{\HtwoCIDupCLo}{\HtwoCIDupCHi}     & \HtwoDetDupG\ci{\HtwoCIDupGLo}{\HtwoCIDupGHi}                 & payload-visible \\
\texttt{silent\_unit\_change}       & \HtwoDetUnitC\ci{\HtwoCIUnitCLo}{\HtwoCIUnitCHi}              & \HtwoDetUnitG\ci{\HtwoCIUnitGLo}{\HtwoCIUnitGHi}              & metadata-borne (gap) \\
\texttt{plausible\_outlier}         & \HtwoDetOutlierC\ci{\HtwoCIOutlierCLo}{\HtwoCIOutlierCHi}     & \HtwoDetOutlierG\ci{\HtwoCIOutlierGLo}{\HtwoCIOutlierGHi}     & metadata-borne (gap) \\
\bottomrule
\end{tabular}
\caption{H2 per-class detection: payload-only critic (C) vs.\ metadata-aware gate (D),
pooled over rate cells (generated from the committed summary). Each cell shows the point
detection rate with its Wilson 95\% interval in percentage points ($n=\HtwoDetN$ corrupted
episodes per class). $^\dagger$~marks the two classes whose empirical detection crosses their
preregistered channel label---the observed boundary is not the assigned one.}
\label{tab:h2}
\end{table}

\subsection{The channel boundary is empirical rather than taxonomic}\label{sec:channel}
The payload/metadata split is often treated as a fixed property of a defect \emph{class}. Our
data shows it is instead a property of \emph{what a given mechanism is shown}. Two registered
labels cross their own boundary, and both directions matter. \texttt{schema\_drift} is
registered payload-visible, yet the realistic payload-only critic misses it (\HtwoDetSchemaC{})
while the deterministic schema predicate catches it (\HtwoDetSchemaG{}): a ``payload-visible''
defect is invisible to a reader that does not implement the right check.
\texttt{superseded\_golden\_record} is registered metadata-borne, yet the critic \emph{does}
see it (\HtwoDetSupersededC{})---not through metadata, but because the current-golden companion
record creates a payload-level contradiction the critic can read. Detectability is therefore
determined by the triple (representation shown, companion evidence available, predicates
implemented), not by the class name. This reframing is what lets the same architecture port to
a new domain without re-deriving a taxonomy (Appendix~\ref{app:analysis}): one moves the
enforcement point to where the signal actually is, then implements the predicate that reads it.
The boundary is drawn by the mechanism, not by the label.

\paragraph{H3 (Gating dominance) --- \ResHthree{} as written.}
The registered Pareto-dominance claim (D dominates B, C, F$(v)$) fails for two honest
reasons: the false-block axis is degenerate (every arm false-blocks at \HthreeFalseBlock{},
so no frontier is traced), and---as anticipated by the arm definition
(\S\ref{sec:gigo})---F$(v)$ benefits from benchmark ground truth and should not be
interpreted as a realistic detection system. Against the realistic payload-only critic C,
the gate D \emph{does} dominate: residual loss \HthreeGateResid{}
\ci{\HthreeGateResidLo}{\HthreeGateResidHi} vs.\ \HthreeCriticResid{}
\ci{\HthreeCriticResidLo}{\HthreeCriticResidHi} (bootstrap 95\% CIs, \S\ref{app:analysis})
and detection \HthreeGateDet{} vs.\ \HthreeCriticDet{} ($2\times$) at equal zero false-block
--- the same metadata-channel advantage as H2. We retain F$(v)$ (rather than drop it) so
this optimistic upper bound stays visible.

\paragraph{H4 (Downstream sufficiency) --- \ResHfour{}.}
The registered target (portfolio recovery $\ge 0.80$) is not met: portfolio recovery is
\HfourRecovery{}. The decomposition is the honest result. Where the gate has a predicate
(freshness), it fully recovers: \texttt{stale\_master\_data} loss falls from
\HfourStaleLossA{} \ci{\HfourStaleLossALo}{\HfourStaleLossAHi} (ungated) to \HfourStaleLossD{}
\ci{\HfourStaleLossDLo}{\HfourStaleLossDHi} (gated; bootstrap 95\% CIs, \S\ref{app:analysis};
the slightly-negative gated value
is within noise of zero per the signed-loss convention of \S\ref{sec:gigo}, not recovery
beyond 100\%). The portfolio number is dominated
by \texttt{silent\_unit\_change} --- the one class carrying a clearly non-zero ungated loss
and exactly the defect the gate has no predicate for (a pre-registered coverage gap, unit
consistency deferred to v1.1). Source-write and lineage predicates (P2/P3) hold by
construction. More episodes cannot move the headline: the gap is structural, not statistical.

\paragraph{Ablations / Tier-2.} Not run in this campaign; the harness and
pre-registrations are in place for a follow-up.

\medskip\noindent
Taken together: the loss-conversion is real and flat across model tiers (H1), and the gate's
value is concentrated on the metadata channel it is designed for --- it beats a realistic
payload-only critic (H2, H3) and fully recovers the defects it covers (H4, freshness) --- while
being transparent about its named coverage gaps. Every figure is generated from logged results
only; no number is hand-entered.

\section{Implications}\label{sec:implications}
Two brief consequences of the three claims; both support M3 rather than extend it.
\paragraph{Economics.} ADR is the action-level micro-foundation of the firm-level
adoption bottleneck. If autonomous value is a $\min(\cdot)$ of complements---capability
\emph{and} trustworthy evidence---then raising capability while evidence quality lags
does not move the $\min$. ADR measures exactly the complement that gating restores.
\paragraph{Governance.} Metadata-borne evidence defects are precisely the data-quality
management and ``examination in view of possible biases'' that EU AI Act Article~10
requires; ISO/IEC 5259~\cite{iso5259} supplies the dimension vocabulary. SARC-DQ's
versioned evidence sets (Prop.~\ref{prop:lineage}) are auditable artifacts for such
obligations. This is the only place we use risk-tier language.

\section{Threats to validity}\label{sec:threats}
\paragraph{Construct validity (are we measuring the right thing?).} Loss is a currency
delta against a same-seed clean counterfactual, so it isolates the corruption effect from
the agent's own decision noise---but it is defined on one realised demand draw per
episode, so single-episode losses are noisy and we report them pooled with bootstrap CIs.
``Silence'' is operationalised as a behavioral doubt-marker AUC and an explicit-flag
fraction; the third intended quantity, an LLM-judge AUC, is deferred, so the silence claim
rests on two of three signals. Materiality uses a fixed threshold ($\tau_m$); we did not
sweep it.
\paragraph{Internal validity (could an artifact explain the result?).} The decider is the
\texttt{policy\_instructed} agent, chosen \emph{after} an initial run to remove a
competence confound (Deviations~\S\ref{sec:deviations}, item~i); both readings are
reported. One re-run was silently truncated by a spend cap and caught in verification
(item~vii); the rate-axis re-runs (H3, H4) and the ladder are on the cap-hardened harness
with zero API errors, and H1-full and H2 (run on the corrected harness before that hardening
was added) were verified complete---all classes ran, no truncation. Seeds, sampler, prompt,
loss definition, and per-run instrumentation are pinned per experiment
(\texttt{EXPERIMENT\_STATUS.md}).
\paragraph{External validity (does it generalise?).} We establish a \emph{mechanism and a
boundary condition, not universal effectiveness} across enterprise workflows. The evidence
is a single decision domain (priced replenishment); corruption is \emph{injected}, not
observed in a live production system; the task has a known clean counterfactual, which is
rarely available in production; the gate \emph{assumes} metadata is present and correct;
predicate coverage is incomplete by design; the ladder uses four Claude tiers from one
vendor, so the flat-across-tiers result is shown within, not across, model families; and
inference-price tier is not itself a scientific measure of capability. The mock GIGO matrix
is a pipeline reference, not evidence about real models. Public base-rate evidence is absent
for \texttt{silent\_unit\_change} and \texttt{plausible\_outlier}: their injector parameters
are declared, flagged defaults, and Tier-2 calibration on public labeled data is future work
(\texttt{benchmarks/gigo/CALIBRATION.md})---surfacing which classes lack a public base rate
is itself part of the calibration contribution. The central supported result is narrow:
absent metadata cannot be recovered by model scaling alone, and a gate with access to the
relevant signal can act where a payload-only mechanism cannot; the magnitude of operational
benefit remains domain-, predicate-, and metadata-quality-dependent.
\paragraph{Statistical-conclusion validity (are the numbers robust?).} Per-class H3/H4
estimates at $n{=}100$ (rate axis) have wide CIs; only \texttt{silent\_unit\_change}
carries a clearly non-zero ungated loss, and we say so rather than over-read the rest.
The H2/H3/H4 headline predictions are reported as \emph{not supported as written}, with
the honest decomposition, not re-fit to pass.

\section{Deviations and clarifications}\label{sec:deviations}
We disclose deviations from the pre-registration so verdicts are read against a
transparent record (full detail in \texttt{reports/FINDINGS.md} and
\texttt{reports/prereg/ADDENDUM-2026-07-09.md}). (i)~\emph{Prompt variant.} The
experiments were run under the \texttt{policy\_instructed} agent (the explicit
newsvendor policy, no data-quality language), pinned after an initial
\texttt{naive}-prompt run returned a uniform-null ADR that reflects decision
incompetence, not data quality; the variant existed in Phase~0b/0c precisely to remove
that confound, and both readings are reported. (ii)~\emph{Loss instrumentation.} Loss
and materiality are measured against the \emph{same agent's} order on the true price
(the paired Phase~0 definition), not against the theoretical optimum; under the old
definition an oracle acting on clean data was scored as heavily defective, i.e.\ the
metric measured the agent's own decision noise rather than corruption. (iii)~\emph{Rate axis.} H1/H2 fix a per-cell corrupted
count (rate as a stratification label) so per-corrupted metrics rest on an equal $n$;
H3/H4 keep true rates, which drive their false-block economics. (iv)~\emph{Silence
sub-claim.} Transcript-based silence metrics were added and logged in the corrected runs:
across the model-tier ladder the behavioral marker AUC ($\le \HoneLadderAuc$) and
explicit-flag fraction (\HoneLadderFlag) both support silence at every rung; the second
behavioral quantity, the LLM-judge AUC, remains deferred (it requires a paid judge turn per
transcript). (v)~\emph{H2 outcome.} The strict channel conjunction is not supported: the
payload/metadata boundary is empirical, the measurement reclassifies \texttt{schema\_drift}
(gate-dominated) and \texttt{superseded\_golden\_record} (payload-detectable), and names two
predicate coverage gaps (\texttt{silent\_unit\_change}, \texttt{plausible\_outlier}) held as
v1.1 future work rather than patched. (vi)~\emph{H3/H4 outcomes.} Both registered headline
predictions are also not supported as written and are reported as measured (Sec.~\ref{sec:results}):
H3's false-block frontier is degenerate and its F$(v)$ arm is a partial oracle, so the fair
comparison is gate-vs-realistic-critic (which the gate wins); H4's portfolio recovery
(\HfourRecovery) is dominated by the \texttt{silent\_unit\_change} coverage gap while the gate
fully recovers the freshness channel it covers. (vii)~\emph{Spend-cap incident and harness
hardening.} One H4 re-run was silently truncated by an API spend cap --- the client swallowed
the transport error and the run committed as apparently-complete (96/96) with a nonsensical
recovery ratio. It was caught in verification, invalidated, and the harness was hardened so a
transport failure is counted and the run aborts-and-resumes rather than committing silent
zeros; the rate-axis re-runs (H3, H4) and the ladder then ran on this hardened harness
(\texttt{instrumentation=api-error-aware-v1}, zero API errors), while H1-full and H2, which
predate the hardening, were verified complete by direct inspection. All predictions are
reported as measured, including where they fail.

To state the integrity boundary plainly: \emph{No threshold, prediction, or verdict
rule was altered at any point in this project; only instrument configuration, under a
dated addendum.} On the silence provenance: \emph{P1 (loss-conversion) is supported under
\texttt{policy\_instructed} and is flat across the model-tier ladder; the behavioral
silence signals (P2 marker AUC, P3 flag fraction) are now measured and support silence at
every rung, while the LLM-judge AUC remains deferred.} On the superseded first wave:
\emph{All first-wave experiment runs are invalid and retained on their branches for audit.}
And on H2: \emph{H2 failed its original registration on the corrected re-run and is reported
as failed; the reframe, not a pass, is the contribution.}

\section{Reproducibility and research transparency}\label{sec:repro}
\paragraph{Commands.} \texttt{make quality} (lint/type/test), \texttt{make smoke}
(Phase 0 mock), \texttt{make gigo-verify} (benchmark drift), \texttt{make status}
(regenerate the experiment manifest), \texttt{make arxiv} (macros + camera-ready source).
Phase 0 config hashes: naive \texttt{c8202a18b58754d8}, policy \texttt{e785bdc87009b84c}.
Per-experiment branch, SHA, prompt/sampler/loss version, and \texttt{config\_hash} are in
\texttt{EXPERIMENT\_STATUS.md} (the single source of truth); seeds, splits, and metric
definitions in \texttt{benchmarks/gigo/SPEC.md}.
\paragraph{Transparency (by construction, not assertion).} We hold this project to the
standard it studies. (i)~\emph{No result is hand-entered}: every printed value is a macro
generated from a committed \texttt{results/<exp>-live} summary. (ii)~\emph{No historical
data were deleted}: invalid first-wave runs and one cap-truncated run remain on their
branches and are listed, with grounds, in \texttt{EXPERIMENT\_STATUS.md}. (iii)~\emph{Every
number is provenance-linked} to a branch, SHA, and \texttt{config\_hash}. (iv)~\emph{Failures
are reported, not hidden}: a garbage-in-garbage-out event in our own parser (\S\ref{sec:pilot})
and a spend-cap truncation (\S\ref{sec:deviations}) are both written up as results.
(v)~\emph{No threshold, prediction, or verdict rule was ever altered}; only instrument
configuration, under a dated addendum.

\section{Conclusion}
We set out three claims: metadata-borne defects are structurally invisible to a
payload-only agent (M1); scaling to higher model tiers does not fix this---silence and
loss-conversion are flat across four model tiers spanning $\approx15\times$ in inference
price (M2); and the remedy is architectural, a metadata-aware Pre-Action Gate at the point
of action (M3). The gate wins on the channel it covers and, where it does not cover a
defect, we report the gap rather than patch it. SARC-DQ does not eliminate upstream
data-quality work and does not detect defects for which no predicate or trustworthy
metadata exists; runtime gating \emph{complements} rather than replaces upstream governance,
and its benefit is predicate- and metadata-quality-dependent. Its contribution is narrower
and architectural: when the discriminating signal exists outside the model's payload
context, reliability requires an enforcement mechanism that can access and act on that
signal. The history of AI safety has largely focused on constraining model \emph{behavior}.
Our results suggest an equally important and largely orthogonal axis: constraining the
quality of the \emph{evidence} on which autonomous decisions are based. As agents move from
answering to acting, runtime evidence governance may become as fundamental to enterprise AI
as memory protection became to operating systems---a cheap, architectural guarantee that a
fast, powerful process cannot act on state it was never allowed to trust.

\paragraph{The flat ladder is not only empirical.} A model-free analytical oracle derived
from the task's decision geometry predicts the conversion of a metadata-borne defect into a
loss \emph{without any model-tier term}---the corrupted price it acts on is fixed by the
injected defect, not by the agent---and this prediction tracks the measured ADR to MAE
\ConvMAE{} (Pearson $r=\ConvR$; the measured interval covers the prediction in
\ConvCICoverage{} cells, Appendix~\ref{app:analysis}). So the ladder's flatness is explained,
not merely observed: capability is absent from the mechanism that sets conversion, while
enforcement placement is the factor that changes it. We claim this only for the tested
architecture and decision setting, not as universal model independence.

\appendix
\section{Analytical companion (deterministic, \$0)}\label{app:analysis}
This appendix strengthens the empirical claims with analytical explanations and robustness
analyses that add \emph{no} new experiment. Every value and figure here is a deterministic,
zero-cost recomputation from the frozen substrate, the frozen injectors, and the committed
result summaries---produced by \texttt{make analysis} (package \texttt{analysis/}, outputs in
\texttt{analysis/out/}, macros in \texttt{generated/analysis.tex}). No threshold, verdict, or
generated experiment value is altered; where an analytical prediction misses a measurement we
report the discrepancy rather than fit it away.

\subsection{The conversion law and its calibration (M2)}
Corollary~\ref{cor:conv} gives a closed-form, model-independent estimate of the conversion
factor $\rho$. Fig.~\ref{fig:conv} calibrates it against the measured agent ADR over the
priced metadata-borne cells (h1-full, arm A): mean absolute error $\ConvMAE$, Pearson
$r=\ConvR$, and the measured 95\% Wald interval covers the oracle prediction in
$\ConvCICoverage$ cells. Because the shown price is fixed by the injected corruption and not by
the agent, model tier is absent from the law by construction---the analytical form of the flat
ladder (M2). The single cell whose interval excludes the prediction is retained in
\texttt{analysis/out/conversion\_law.json}, not discarded.

\begin{figure}[t]\centering
\includegraphics[width=0.52\linewidth]{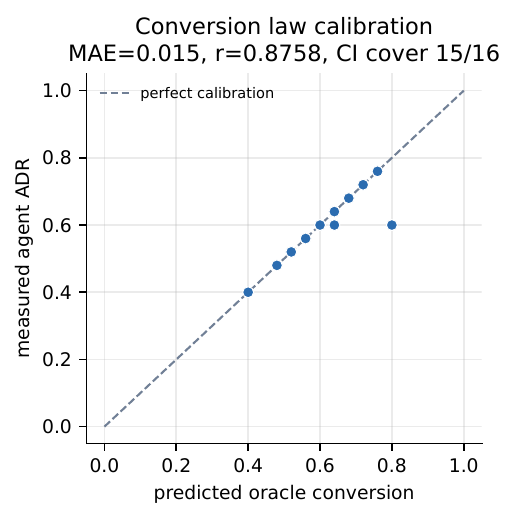}
\caption{Predicted oracle conversion vs.\ measured agent ADR. Proximity to the diagonal is
the calibration; the agent tracks the closed-form, model-free prediction.}
\label{fig:conv}
\end{figure}

\subsection{H4 as a coverage-accounting identity}
Portfolio recovery decomposes exactly as
$\text{recovery}=(L_A-L_D)/(L_A-L_E)$ with $L_x=\sum_c w_c\,\ell_{x,c}$, class weight
$w_c=n_c/N$, and $L_E=0$. Recomputed from the committed h4-recovery losses this identity
reproduces the reported portfolio recovery ($\CovRecovery$) to a residual of $\CovResidual$.
The decomposition (Fig.~\ref{fig:cov}) explains \emph{why} the headline is low without any
appeal to statistics: the single uncovered class $\CovDominant$ carries $\CovDominantSharePct$\%
of the recoverable denominator, so a gate that (by design) has no predicate for it cannot move
the portfolio number---exactly the pre-registered coverage gap.

\begin{figure}[t]\centering
\includegraphics[width=0.72\linewidth]{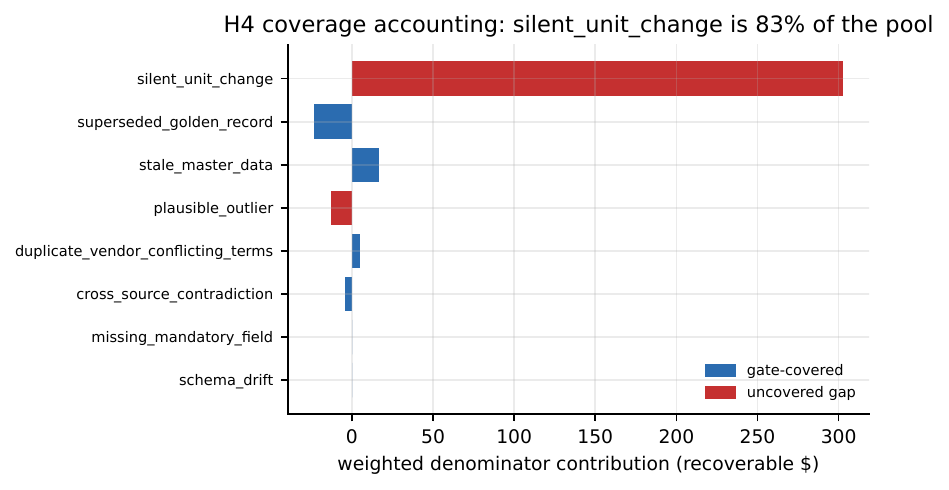}
\caption{Weighted denominator (recoverable \$) by class; blue = gate-covered, red = uncovered
gap. One uncovered, high-magnitude class dominates the recoverable pool.}
\label{fig:cov}
\end{figure}

\subsection{Formal coverage matrix}
Table~\ref{tab:coverage} states, per class, where the discriminating signal lives, whether each
mechanism can observe it, the implemented v1 predicate, and the remaining limitation---the
concise map of what the architecture does and does not cover.

\begin{table}[t]\centering\small
\setlength{\tabcolsep}{4pt}
\begin{tabular}{@{}lccccl@{}}
\toprule
class & signal & critic & gate & predicate & remaining limitation \\
\midrule
\texttt{stale\_master\_data}        & metadata & no  & yes & freshness    & needs a fresh clock \\
\texttt{superseded\_golden\_record} & meta/payload & yes & yes & golden-unique & needs version/companion \\
\texttt{schema\_drift}              & payload  & no  & yes & schema       & needs the typed check \\
\texttt{missing\_mandatory\_field}  & payload  & yes & yes & complete     & --- \\
\texttt{cross\_source\_contradiction} & payload & yes & yes & consistency & needs both sources keyed \\
\texttt{duplicate\_vendor\_conflicting\_terms} & payload & yes & no & (consistency) & alt-keyed rows not compared \\
\texttt{silent\_unit\_change}       & none printed & no & no & --- & v1.1 unit-consistency gap \\
\texttt{plausible\_outlier}         & lineage only & no & no & --- & v1.1 outlier gap \\
\bottomrule
\end{tabular}
\caption{Coverage matrix (frozen gate design). ``critic''/``gate'' record whether the
payload-only critic / metadata-aware gate can observe the defect; the last two classes are the
declared, unpatched gaps.}
\label{tab:coverage}
\end{table}

\subsection{Statistical reporting: row-level uncertainty}
\emph{Ladder flatness (H1).} The ``flat to slightly rising'' claim is quantified from the
committed h1-ladder pooled counts over the metadata-borne classes: the top-minus-bottom
endpoint difference is $\LadderEndpointDiff$ with a two-proportion 95\% Wald interval
$[\LadderEndpointCILo,\LadderEndpointCIHi]$ that includes zero, and the ordinary-least-squares
trend is $\LadderTrendPerTier$ ADR per tier step. Flatness is thus a measured quantity with an
interval, not an adjective.

\emph{H2 detection intervals.} Table~\ref{tab:h2} reports a Wilson 95\% interval for every
per-class detection rate, for both the critic (C) and the gate (D). The detected count is
recovered exactly per rate cell (each stored rate is an integer multiple of $1/25$), never by
multiplying a rounded pooled rate, and pooled to $n=\HtwoDetN$ corrupted episodes per class.

\emph{H3/H4 residual-loss intervals.} The residual-loss quantities in \S\ref{sec:results} carry
nonparametric bootstrap 95\% intervals: $\BootReplicates$ percentile resamples (fixed seed
$\BootSeed$) over the committed per-corrupted-episode paired losses, the resampling unit. The
signed paired-loss convention is preserved (losses are never clipped), the point estimate is
left at the committed pool mean, and an interval that crosses zero---such as the gated
\texttt{stale\_master\_data} loss---is reported as consistent with zero, not as recovery beyond
the measured value.

\subsection{Robustness: leave-one-class-out, ablations, and threshold sensitivity}
Removing the dominant uncovered class $\LooTopClass$ from the portfolio moves recovery from
$\LooFull$ to $\LooTopRecovery$ (Fig.~\ref{fig:loo})---a structural demonstration that this one
class, not statistical noise, sets the headline. Predicate ablations
(\texttt{analysis/out/robustness.json}) confirm that no v1 predicate subset recovers the
portfolio, because none covers that class. Threshold sensitivity (Fig.~\ref{fig:thresh}) sweeps
the materiality threshold on the oracle conversion---a faithful proxy, given the $\ConvMAE$
calibration above---from a mean ADR of $\ThreshLoAdr$ down to $\ThreshHiAdr$; the registered
$\tau_m=\ThreshReg$ (ADR $\ThreshRegAdr$) is one point on a smooth curve, not a cliff edge, so
the verdicts do not hinge on the threshold choice.

\begin{figure}[t]\centering
\includegraphics[width=0.70\linewidth]{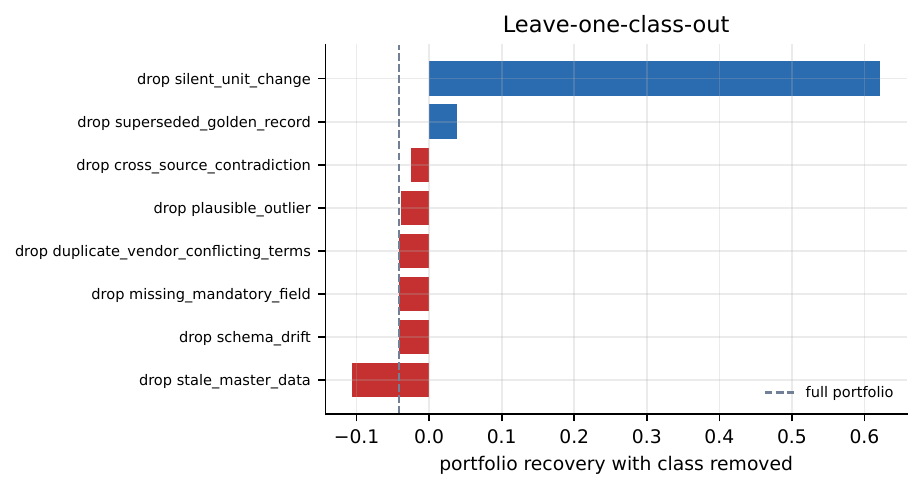}\\[4pt]
\includegraphics[width=0.46\linewidth]{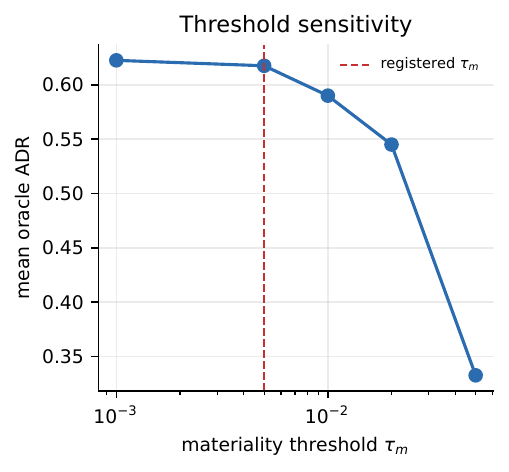}
\caption{Top: portfolio recovery with each class removed (dashed = full portfolio). Bottom:
oracle ADR vs.\ the materiality threshold; the registered $\tau_m$ is marked.}
\label{fig:loo}\label{fig:thresh}
\end{figure}

\subsection{Stress: dependence on metadata quality}
The gate reads metadata, so its detection is only as good as the metadata it is handed.
Fig.~\ref{fig:stress} re-runs the real gate on the covered classes under deterministic
metadata-degradation operators. Pooled detection falls from $\StressIntact$ (intact) to
$\StressAllDegraded$ (all metadata degraded); the most fragile class is $\StressFragile$, whose
freshness signal collapses when the clock is erased. This is an honest boundary: placement helps
to the extent the metadata channel is intact.

\begin{figure}[t]\centering
\includegraphics[width=0.72\linewidth]{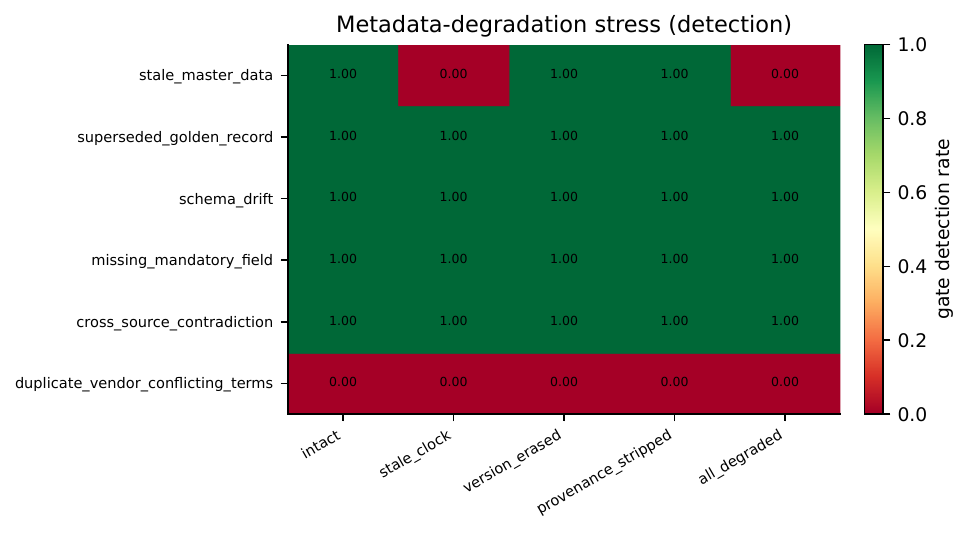}
\caption{Gate detection on covered classes under metadata degradation. Green survives; red
collapses. Detection is a function of metadata quality, not a universal guarantee.}
\label{fig:stress}
\end{figure}

\subsection{Falsification: it is targeting, not frequency}
A natural null is ``the gate just intervenes more.'' We falsify it (Fig.~\ref{fig:fals}) by
comparing, at the \emph{same} intervention budget, the real gate against corruption-blind
controls over the full population. The gate recovers $\FalsGate$ of the positive-loss pool
versus $\FalsRandom$ for a frequency-matched random policy, $\FalsShuffled$ for permuted flags,
and $\FalsAnti$ for an adversarially mis-targeted policy (gate advantage $\FalsGap$). Recovery
tracks alignment with the corruption signal, not how often one intervenes.

\begin{figure}[t]\centering
\includegraphics[width=0.48\linewidth]{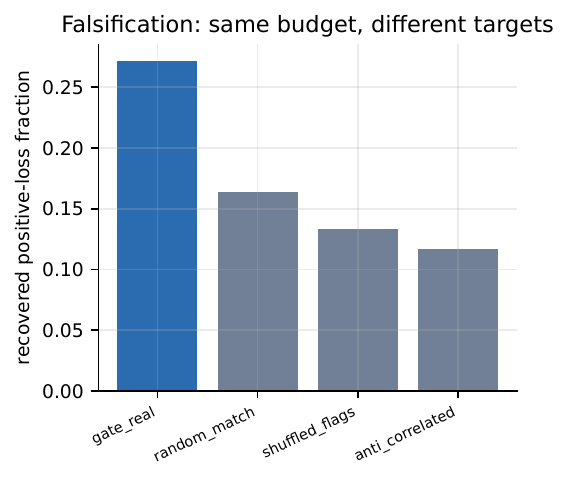}
\caption{Recovered positive-loss fraction at equal intervention budget. The metadata gate (blue)
beats every corruption-blind control; the benefit is targeting, not frequency.}
\label{fig:fals}
\end{figure}

\subsection{A second decision process: architectural portability}
To show the architecture ports---not that it wins universally---we instantiate a structurally
different decision: a B2C personalised-promotion eligibility task (a binary offer, priced and
paired, not a newsvendor order). The same payload/metadata split and the \emph{same} predicate
family (\texttt{freshness}, \texttt{golden\_record\_unique}, \texttt{complete}) apply. Six
metadata-borne promo corruptions (expired campaign, superseded segmentation, stale score,
outdated rules, missing consent, obsolete recommendation) yield a portfolio ADR of $\SecondADR$,
and the reused gate recovers $\SecondRecovery$ of the induced loss (Fig.~\ref{fig:second}). The
ADR profile is domain-specific---eligibility-borne defects (expired campaign $\SecondExpiredADR$,
missing consent $\SecondConsentADR$) dominate here, unlike \texttt{silent\_unit\_change} in
procurement---which is the point: the mechanism transfers even though the loss surface does not.

\begin{figure}[t]\centering
\includegraphics[width=0.74\linewidth]{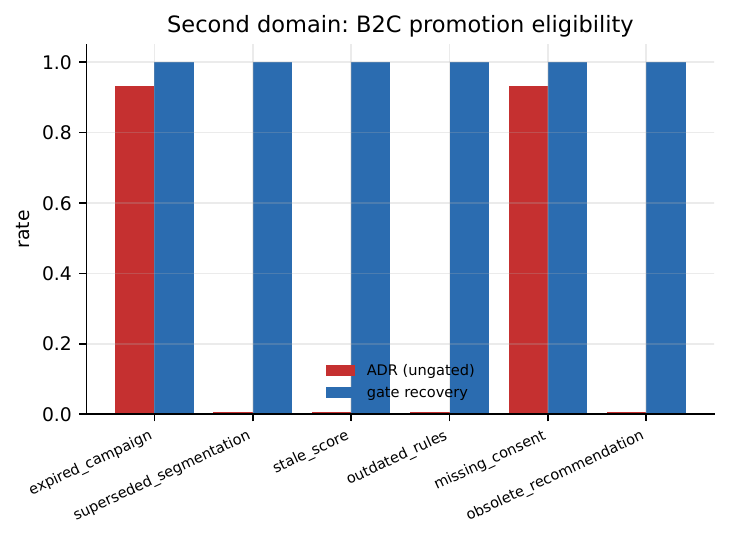}
\caption{Second domain (B2C promotion eligibility): per-class ungated ADR (red) and gate
recovery (blue) using the identical predicate family. The architecture ports; performance is
domain-specific.}
\label{fig:second}
\end{figure}

\end{document}